\documentclass[conference]{IEEEtran}
\usepackage{enumerate}
\usepackage{amssymb,stmaryrd,amsmath,amsfonts,rotating}
\usepackage[noadjust]{cite} \usepackage{color} 
\usepackage[vflt]{floatflt} \usepackage{epic}
\usepackage{color}
\usepackage{verse}
\usepackage[mathcal]{euscript}
\usepackage{mathrsfs}
\usepackage{dsfont}

\newcommand{\ent}{\ensuremath{{\tt{h}}}}

\RequirePackage{bbm}
\definecolor{darkgreen}{rgb}{0.14,0.5,0.14}

\newtheorem{theorem}{Theorem}
\newcommand{\btheo}{\begin{theorem}}
\newcommand{\etheo}{\end{theorem}}
\newcommand{\bproof}{\begin{proof}}
\newcommand{\eproof}{\end{proof}}
\newtheorem{definition}[theorem]{Definition}
\newcommand{\bdefi}{\begin{definition}}
\newcommand{\edefi}{\end{definition}}
\newtheorem{fact}[theorem]{Fact}
\newcommand{\bprop}{\begin{fact}}
\newcommand{\eprop}{\end{fact}}
\newtheorem{corollary}[theorem]{Corollary}
\newcommand{\bcor}{\begin{corollary}}
\newcommand{\ecor}{\end{corollary}}
\newtheorem{example}[theorem]{Example}
\newcommand{\bex}{\begin{example}}
\newcommand{\eex}{\end{example}}
\newtheorem{lemma}[theorem]{Lemma}
\newcommand{\blemma}{\begin{lemma}}
\newcommand{\elemma}{\end{lemma}}
\newtheorem{remark}[theorem]{Remark}
\newcommand{\bremark}{\begin{remark}}
\newcommand{\eremark}{\end{remark}}
\newtheorem{conj}[theorem]{Conjecture}
\newcommand{\bconj}{\begin{conj}}
\newcommand{\econj}{\end{conj}}

\newcommand{\naturals}{\ensuremath{\mathbb{N}}}

\newcommand{\expectation}{\ensuremath{\mathbb{E}}}

\newcommand{\prob}{\ensuremath{\mathbb{P}}}

\def\0{{\tt 0}} 
\def\1{{\tt 1}} 
\def\?{{\tt *}} 
\newcommand{\dee}{{\text d}}
\newcommand{\BPsmall}{\ensuremath{\text{\tiny BP}}} 
\newcommand{\qed}{{\hfill \footnotesize $\blacksquare$}}
\renewcommand{\mid}{\,|\,}

\newcommand{\dens}[1]{\mathsf{#1}}

\newcommand{\Ldens}[1]{\dens{#1}}

\newcommand{\Ddens}[1]{\mathfrak{{#1}}}

\newcommand{\BMS}{\ensuremath{\text{BMS}}}

\newcommand{\absDdist}[1]{\absd{\mathfrak{\MakeUppercase{#1}}}}
\newcommand{\absDdens}[1]{\absd{\Ddens{#1}}}

\newcommand{\absd}[1]{|#1|}

\newcommand{\dr}{d_r}
\newcommand{\dl}{d_l}

\DeclareMathOperator{\perr}{\mathfrak{E}}

\newcommand{\ind}{\mathbbm{1}}
\newcommand{\indicator}[1]{\ind_{\{ #1 \}}}

\DeclareMathOperator{\batta}{\mathfrak{B}}      
\newcommand{\entropy}{\text{H}}

\newcommand{\Lip}{\text{Lip}}

\newcommand{\vconv}{\circledast}
\newcommand{\cconv}{\boxast}

\newcommand{\SatLdens}[1]{\lfloor \Ldens{#1}\rfloor}
\newcommand{\sym}[1]{{#1_{\text{\tiny sym}}}}
\newcommand{\SatabsDdist}[2]{\absd{\lfloor\mathfrak{\MakeUppercase{#1}}\rfloor_{#2}}}
\newdimen\arrayruleHwidth
\makeatletter
\def\Hline{\noalign{\ifnum0=`}\fi\hrule \@height \arrayruleHwidth
   \futurelet \@tempa\@xhline}
	\makeatother
\allowdisplaybreaks
\begin{document}
\title{The Effect of Saturation on Belief Propagation Decoding of LDPC Codes} 
\author{\IEEEauthorblockN{Shrinivas Kudekar\IEEEauthorrefmark{1}, Tom Richardson\IEEEauthorrefmark{1} and Aravind Iyengar\IEEEauthorrefmark{1}\\ }
\IEEEauthorblockA{\IEEEauthorrefmark{1}Qualcomm,  USA}
 }

\maketitle

\begin{abstract}
We consider the effect of LLR saturation on belief propagation decoding of low-density parity-check codes. 
Saturation occurs universally in practice and is known to have a significant effect on error floor performance.
Our focus is on threshold analysis and stability of density evolution.

We analyze the decoder for certain low-density parity-check code ensembles and show that 
belief propagation decoding generally degrades gracefully with saturation.
Stability of density evolution is, on the other hand, rather strongly affected by saturation and
the asymptotic qualitative effect of saturation is similar to reduction of variable node degree by one.
\end{abstract}

\section{Introduction} Standard belief propagation (BP) decoding of binary LDPC codes involves
passing messages typically representing log-likelihood ratios (LLRs) 
which can take any value in $\overline{\mathbb{R}} \triangleq
\mathbb{R}\cup \{\pm\infty\}$ \cite{RiU08}. 
Practical implementations of decoders typically use uniformly quantized and bound LLRs.
Hence, it is of interest to understand the effect of saturation of LLR magnitudes
as a perturbation of full belief propagation.
We call such a clipped or a saturated decoder
as a {\em saturating} belief propagation decoder (SatBP).
Note that the decoder is strictly speaking not a BP decoder, but we adhere to the BP
nomenclature as we view SatBP as a perturbation of BP.

The papers \cite{6290251,6120169,5485006, 6567866} consider the effect of saturation on error floor performance.  It is observed in these works that saturation can limit the ability of decoding to escape trapping set behavior, thereby worsening error floor performance.
Although we take a different approach in this paper by focussing on asymptotic behavior, the underlying message is similar: saturation can dramatically effect the stability of the decoder.
In \cite{6284049,6685976} some decoder variations are given the help reduce error floors.
Here we see an explicit effort to ameliorate the effect of saturation.
A related but distinct direction was taken in \cite{VNC14}.  There the authors made modifications to discrete node update rules so as to reduce error floor failure events. 
There have been other work that examine the effects of practical concessions.
In \cite{6134417} the authors consider the effect of quantization in LDPC coded flash memories.
In \cite{5624882, 6363268} model the effect of saturation and quantization as noise terms.
Finally, in \cite{5205826} an analysis was done to evaluate the effect on capacity on quantization of
channel outputs.

In the design of capacity-achieving codes it would be helpful to understand
how practical decoder concessions, like saturation, affect performance.
In particular, if LLRs are saturated at magnitude $K$ then how much degradation
from the BP threshold should be expected.  
Naturally, one expects that as $K\to +\infty$, that one can reliably transmit
arbitrarily close to the BP threshold \cite{RiU08}.
We will see that this is not entirely correct and that, in particular, saturation can
undermine the stability of the perfect decoding fixed point. 

\section{BP decoding,  Density Evolution and the Wasserstein Distance} In this
section we briefly review the BP decoder and the density
evolution analysis \cite{RiU01} in the case of transmission over a
general BMS channel. Most of the material presented here can be
found in \cite{RiU08}.

Let $X (=\pm 1)$ denote the channel input, let $Y$
denote the channel output, and
let $p(Y=y \mid X=x)$ denote the {\em transition probability}
of the channel. We generally characterize
the channel by its so-called $L$-density, $\Ldens{c}$ which 
is the distribution of
$ \ln \frac{p(Y \mid X=1)}{p(Y \mid X=-1)} $
conditioned on $X=1.$  Generally, we may assume that 
$
Y=\ln \frac{p(Y \mid X=1)}{p(Y \mid X=-1)}\,,
$
i.e., the output of the channel is the associated LLR.
Channel symmetry is the condition $p(Y=y \mid X=x) = p(Y=-y \mid X=-x)$
and the $L$-densities $\Ldens{c}$ that result are symmetric, \cite{RiU08},
which means $e^{-\frac{1}{2}x}\Ldens{c}(x)$ is an even function of $x.$
 We recall that all densities
which stem from BMS channels are symmetric, see \cite[Sections
4.1.4, 4.1.8 and 4.1.9]{RiU08}. \qed 

Given $Z$ distributed according to $\Ldens{c}$, we write $\Ddens{c}$
to denote the distribution of $\tanh(Z/2),$
and $\absDdens{c}$
to denote the distribution of $|\tanh(Z/2)|.$
We refer to these as the $D$ and $|D|$ distributions.
We use $\absDdist{c}$ to denote the corresponding cumulative $|D|$ distribution,
see \cite[Section~4.1.4]{RiU08}. 
Under symmetry the distribution of $|Z|$ determines the distribution of $Z.$

For threshold analysis of LDPC ensembles we typically consider a
parameterized {\em family} of channels and write $\{ \BMS(\sigma)\}$ to
denote the family as parameterized by the scalar $\sigma$. Often it
will be more convenient to denote this family by $\{\Ldens{c}_\sigma\}$,
i.e., to use the family of $L$-densities which characterize the
channel family. 
One natural candidate for the parameter $\sigma$ is the 
entropy of the channel denoted by $\ent$. Thus, we also consider the characterization
of the family given by the BMS($\ent$). 

Let $p_{Z \mid X}(z\mid x)$ denote the transition probability
associated to a BMS channel $\Ldens{c}'$ and let $p_{Y \mid X}(y
\mid x)$ denote the transition probability of another BMS channel
$\Ldens{c}$.  We then say that $\Ldens{c}'$ is {\em degraded} with
respect to $\Ldens{c},$ denoted $\Ldens{c} \prec \Ldens{c}',$
if there exists a channel $p_{Z \mid Y}(z\mid
y)$ so that  $p_{Z \mid X}(z \mid x) = \sum_{y} p_{Y \mid X}(y\mid
x) p_{Z \mid Y}(z\mid y)$.


A BMS channel family $\{\BMS(\sigma)\}_{\underline{\sigma}}^{\overline{\sigma}}$ is said to be {\em ordered} (by
degradation) if $\sigma_1 \leq \sigma_2$ implies $\Ldens{c}_{\sigma_1}
\prec \Ldens{c}_{\sigma_2}$.

Useful functionals of densities include the Battacharyya, the entropy, and
the error probability functionals.
For a density $\Ldens{a}$, these are denoted by $\batta(\Ldens{a})$,
$\entropy(\Ldens{a})$, and $\perr(\Ldens{a})$, respectively and are defined by
\begin{align*}
\batta(\Ldens{a}) & =\expectation( e^{-Z/2}), \;\;
\entropy(\Ldens{a})  = \expectation( \log_2(1\!+\!e^{-Z})),  \\
\perr(\Ldens{a}) & = \prob{\{Z<0\}} + \frac{1}{2} \prob \{Z=0\}
\end{align*}
where $Z$ is distributed according to $\Ldens{a}$.
Note that these definitions are valid even if $\Ldens{a}$ is not symmetric, although they
loose some of their original meaning.  We will apply these definitions, especially $\batta,$ to saturated densities
that are not necessarily symmetric.

%

\subsubsection{BP Decoder and Density Evolution}

The definition of the standard BP decoder can be found in \cite{RiU08}. The
asymptotic performance of the BP decoder is given by the density evolution
technique \cite{RiU01, RiU08}.  

\begin{definition}[Density Evolution for BP Decoder cf. \cite{RiU08}]
For $\ell \geq 1$, the DE equation for a
$(\dl, \dr)$-regular ensemble is given by
$$
T(\Ldens{c}, \Ldens{x}) \triangleq \Ldens{x}_{\ell} = \Ldens{c} \vconv (\Ldens{x}^{\cconv \dr-1}_{\ell-1})^{\vconv \dl-1}.
$$
Here, $\Ldens{c}$ is the $L$-density of the BMS channel over which
transmission takes place and $\Ldens{x}_{\ell}$ is the density emitted by
variable nodes in the $\ell$-th round of density evolution. Initially we
have $\Ldens{x}_{0}=\Delta_0$, the delta function at $0$. The operators
$\vconv$ and $\cconv$ correspond to the convolution of densities at
variable and check nodes, respectively, see \cite[Section 4.1.4]{RiU08}.
\qed
\end{definition}
{\em Discussion}: The DE analysis is simplified when we consider the class of
symmetric message-passing decoders. The definition of symmetric message-passing
decoders can be found in \cite{RiU08}. Note that this definition of symmetry
pertains to the actual messages in the decoder and not to the densities which
appear in the DE analysis. We will see later that the clipped or the saturated
decoder is a symmetric message-passing decoder in this sense and hence its DE analysis is
simplified by restricting to consideration of the all-zero codeword.

\begin{definition}[BP Threshold for regular ensembles]
Consider an ordered and complete channel family $\{\Ldens{c}_\ent\}$.
Let $\Ldens{x}_{\ell}(\ent)$ denote the distribution in the $\ell$-th round
of DE when the channel is $\Ldens{c}_\ent$.
Then the {\em BP threshold} of the $(\dl, \dr)$-regular ensemble 
is defined as
$
\ent^{\BPsmall}(\dl, \dr, \{\Ldens{c}_\ent\})  = \sup\{\ent: \Ldens{x}_{\ell}(\ent) \stackrel{\ell \to \infty}{\rightarrow} \Delta_{+\infty}\}.
$
Under symmetry an equivalent definition is
\begin{align*}
\ent^{\BPsmall}(\dl, \dr, \{\Ldens{c}_\ent\}) & = \sup\{\ent: \perr(\Ldens{x}_{\ell}(\ent)) \stackrel{\ell \to \infty}{\rightarrow} 0\}.
\end{align*}
The later form is more convenient for our purposes and it is the one we shall adopt.
\qed
\end{definition}

In the sequel we will use the Wasserstein metric to measure distance between distributions.
We recall the definition of the Wasserstein metric below. For details  see \cite{KRU12b}.
\begin{definition}[Wasserstein Metric -- \protect{\cite[Chapter 6]{Villani09}}]\label{def:wasserstein}
Let $\absDdens{a}$ and $\absDdens{b}$ denote two $|D|$-distributions.
The Wasserstein  metric,
denoted by $d(\absDdens{a}, \absDdens{b})$, is defined as
\begin{align}\label{eq:blmetric}
d(\absDdens{a}, \absDdens{b})=\!\!\!\!\!\!\sup_{f(x) \in \Lip(1)[0, 1]} \!\Big\vert \int_{0}^{1} \!\!f(x)(\absDdens{a}(x)\!-\! \absDdens{b}(x)) \,\dee x \Big\vert,
\end{align}
\qed
\end{definition}

%

\section{Saturated Belief Propagation Decoding}

We first consider the analysis of the saturated BP decoder. More precisely, 
we consider decoding with BP rules but with messages restricted to the domain
$[-K, K]$ for some $K > 0$. 

\subsection{Saturated Decoder}
\begin{definition}[Saturation]
We define the \emph{saturation} operation at $\pm K$ for some $K \in \mathbb{R}^+$, denoted $\lfloor\cdot\rfloor_K$, by
\begin{equation} \label{eq_clip}
\lfloor x \rfloor_K = \min(K, |x|)\cdot\mathrm{sgn}(x).
\end{equation}
\end{definition}

\begin{definition}[SatBP Decoder]\label{def:clippedBPdecoder}
Consider the standard $(\dl, \dr)$-regular ensemble. 
The saturated BP decoder is defined by the following rules. 
Let $\phi^{(\ell)}(\mu_1,\dots,\mu_{\dr-1})$ and $\psi^{(\ell)}(\mu_1, \dots,
\mu_{\dl-1})$ denote the outgoing message from the
check node and the variable node side respectively. Abusing the notation above 
 and denoting the incoming messages on both the check node and the variable node side by $\mu_i$, we have
\begin{align*}
\phi^{(\ell)}(\mu_1,\dots,\mu_{\dr-1}) & = \Big\lfloor
2\tanh^{-1}\Big(\prod_{i=1}^{\dr-1}\tanh(\mu_i/2)\Big)\Big\rfloor_K, \\
\psi^{(\ell)}(\mu_1,\dots,\mu_{\dl-1}) & = \Big\lfloor \mu_0 + \sum_{i=1}^{\dl-1}
\mu_i \Big\rfloor_K,
\end{align*}
where $\mu_0$ is the message coming from the channel. Also, we set
$\phi^{(0)}(\mu_1,\dots,\mu_{\dr-1}) = 0$. \qed
\end{definition}

Using Definition 4.83 in \cite{RiU08} we have the following.
\begin{lemma}[SatBP Decoder is symmetric]\label{lem:symmClipped}
The SatBP decoder given in Definition~\ref{def:clippedBPdecoder} is a symmetric decoder.
\end{lemma}

{\em Discussion:} The symmetry of the message-passing decoder together with symmetry of the channel allows us to use the all-zero codeword assumption. This along with the concentration results (see
Theorem 4.94 in \cite{RiU08}) allows to write down the density evolution of the
SatBP decoder in the usual way.

Given $X \sim \Ldens{a}$ let $\SatLdens{a}_K$ denote the distribution of 
$\lfloor X \rfloor_K.$  Note that the saturation operation can be viewed as a channel taking
$X$ to $\lfloor X \rfloor_K.$  We immediately have
$
\Ldens{a} \prec \SatLdens{a}_K\,.
$
In general $\SatLdens{a}_K$ will not be symmetric even if $\Ldens{a}$ is symmetric
since we will not typically have $\SatLdens{a}_K(-K)= e^{-K} \SatLdens{a}_K(K).$
If $\Ldens{a}$ is symmetric then we will have
$\SatLdens{a}_K(-K)  \le  e^{-K} \SatLdens{a}_K(K).
$
Although using lemma~\ref{lem:symmClipped} one can write down the DE recursion
for the SatBP decoder, we know that in general the densities will not be
symmetric. Two of the most useful properties of DE for BP is that it preserves both symmetry of 
densities and ordering by degradation.  These properties are sacrificed by saturation, but can be recovered
with a slight variation. 
The idea is to slightly degrade the density by moving some probability mass from
$K$ to $-K.$  This can be interpreted operationally as flipping the sign of a message with magnitude
$K$ with some probability $\lambda.$  The flipping rate $\lambda$ is chosen so that the resulting 
probability that the sign of the message is incorrect is $e^{-K}/(1+e^{-K}).$  In general $\lambda$
is upper bounded by this value and for large $K$ this is a small perturbation. 
With this perturbation both density symmetry and ordering by degradation are recovered.
Let us introduce the notation $D(p,z)$ to denote the density
$
D(p,z) = p \Delta_{-z} + (1-p) \Delta_{z}\,.
$
Using this notation we have for symmetric $\Ldens{a},$
$
\SatLdens{a}_K = \gamma D(q,z)(x) + \Ldens{a}(x)\mathds{1}_{\{|x| < K\}}
$
where $\gamma = \prob_{\Ldens{a}} \{ |x| \ge K \}$ and $\gamma q = \prob_{\Ldens{a}} \{ x \le -K \}.$
\begin{lemma}[Symmetric Saturation] \label{lem:symclipprop}
Given a symmetric density $\Ldens{a}$ we define
$
\SatLdens{a}_\sym{K} = \gamma D(p,z)(x) + \Ldens{a}(x)\mathds{1}_{\{|x| < K\}}
$
where $p = e^{-K}/(1+e^{-K})$ and $\gamma =  \prob_{\Ldens{a}} \{ |x| \ge K \}.$
Then, (i) $\SatLdens{a}_\sym{K}$ is a symmetric $L$-density and (ii) $\SatLdens{a}_K \prec \SatLdens{a}_\sym{K}$.
\end{lemma}
\begin{proof}
Part (i) is immediate. To prove part (ii) we note that comparing with the unsymmetrized case we see that 
\( p \ge q \,.\)  Thus, $\SatLdens{a}_\sym{K}$ can be realized by taking messages with distribution
$\SatLdens{a}_K$ and flipping the sign of a message with magnitude $K$ by a quantity $\lambda$
 which is determined by
$
p = \frac{e^{-K}}{1+e^{-K}} = \lambda (1-q) + (1-\lambda) q\,.
$
\end{proof}
As a consequence of Lemma \ref{lem:symclipprop}, we will term the operation used
to obtain $\SatLdens{a}_\sym{K}$ from $\Ldens{a}$ as \emph{symmetric-saturation}. 

We summarize all the claims above in the following.
\begin{corollary}\label{lem:degradationOrder}
We have $\Ldens{a} \prec \SatLdens{\Ldens{a}}_K \prec \SatLdens{a}_\sym{K}.$ 
\end{corollary}

\begin{lemma}\label{lem:DistSymClip}
Let $\Ldens{a}$ be a symmetric $L$-density. 
Then,
$d(\Ldens{a}, \SatLdens{a}_\sym{K}) \leq 1 - \tanh(K/2).
$
\end{lemma}
\begin{proof}
For any $0 \le z < K$ we have 
\(
\prob_{\Ldens{a}} \{ x \le z \}
=
\prob_{\SatLdens{a}_K} \{ x \le z \}
{=}
\prob_{\SatLdens{a}_\sym{K}} \{ x \le z \}
\)
and for any $z \ge K$ we have 
\(
1
=
\prob_{\SatLdens{a}_K} \{ x \le z \}
=
\prob_{\SatLdens{a}_\sym{K}} \{ x \le z \}\,.
\)
Since $\tanh(x/2)$ is increasing and $\tanh(-x/2)=-\tanh(x/2)$ we have
$
\SatabsDdist{a}{\sym{K}} (z) =
\indicator{z<\tanh(K/2)} \absDdist{a} (z)
+
\indicator{z\ge\tanh(K/2)}\,.
$
By \cite{KRU12b}, we have that the Wasserstein distance is equivalent to the
$L_1$ norm of the difference between the $|D|$-distribtions.
Clearly, the distance is
bounded by $1 - \tanh (K/2)$.
\end{proof}
\begin{definition}[DE for Sym. and Non-Sym. Saturation]
The DE for non-symmetric saturation decoder is 
$S_\sym{K}(\Ldens{c}, \Ldens{x}) = \left\lfloor T(\Ldens{c}, \Ldens{x})\right\rfloor_\sym{K}$ and 
 for the non-symmetric saturation decoder is 
$S_K(\Ldens{c}, \Ldens{x}) = \left\lfloor T(\Ldens{c}, \Ldens{x})\right\rfloor_{K}$. \qed
\end{definition}

We now estimate the distance between the densities appearing in the DE of usual
BP and the DE of the symmetric-saturation operation.  
\begin{lemma}[Distance Between Symmetric-SatBP and BP]\label{lem:distsymclippedandBP}
Consider $\ell$ iterations of the forward DE for the usual BP and the Symmetric-Saturation operation.
Then
\begin{align*}
\batta(S_\sym{K}^{(\ell)}&(\Ldens{c}, \Delta_0)) \leq  \batta(T^{(\ell)}(\Ldens{c}, \Delta_0)) \!+\! 2\sqrt{2}e^{\frac{\!-\!K \!+\! \ell\cdot\ln (2(\dl\!-\!1)(\dr\!-\!1))}2}.
\end{align*}
\end{lemma}

\begin{proof}
Using the triangle inequality, (viii), Lem. 13 in \cite{KRU12b} and  Lemma~\ref{lem:DistSymClip} 
we obtain the upper bound $d(T^{(\ell)}(\Ldens{c}, \Delta_0), S_\sym{K}^{(\ell)}(\Ldens{c}, \Delta_0)) \leq 
\alpha_{\ell} d(T^{(\ell - 1)}(\Ldens{c}, \Delta_0), S_\sym{K}^{(\ell - 1)}(\Ldens{c}, \Delta_0)) \!+\! 1 \!-\! \tanh\Big(\frac{K}{2}\Big)$,
where
$
\alpha_{\ell}  = 2 (\dl-1)  
 \sum_{j\!=\!1}^{\dr\!-\!1} (1\!-\!\batta^2(\Ldens{a}))^{\frac{\dr\!-\!1\!-\!j}2}(1\!-\!\batta^2(\Ldens{b}))^{\frac{j\!-\!1}2}
$, $\Ldens{a} = T^{(\ell - 1)}(\Ldens{c}, \Delta_0)$ and $\Ldens{b} = S_\sym{K}^{(\ell - 1)}(\Ldens{c}, \Delta_0)$.
Continuing with the above inequality we get the upper bound,
 $(1\! - \!\tanh\Big(\frac{K}{2}\Big))(1 \!+\! \alpha_{\ell} \!+\! \alpha_{\ell}\alpha_{\ell-1} \!+\! \dots \!+\! \alpha_{\ell}\alpha_{\ell-1}\cdots \alpha_2).
$
In general, we are transmitting below the BP threshold, so both $\Ldens{a}$ and $\Ldens{b}$ could 
be close to $\Delta_{+\infty}$. 
Thus, $(1 \!+\! \alpha_{\ell} \!+\! \alpha_{\ell}\alpha_{\ell-1} \!+\! \dots \!+\! \alpha_{\ell}\alpha_{\ell-1}\cdots \alpha_2) \leq  (2(\dl-1)(\dr-1))^{\ell}$. Using $1 - \tanh(K/2) \leq 2e^{-K}$ and (ix) Lemma~13 
in \cite{KRU12b} we get the lemma.
\end{proof}

\section{Convergence of Nonsymmetric Saturated DE}
In the previous section we show that, when transmitting below the threshold of the full BP decoder, the Battacharrya parameter of the densities in the symmetric SatBP decoder can be arbitrarily small. In this section we make the connection to the non-symmetric SatBP.

\subsection{Non-symmetrized SatBP Decoder}
We now show that the Battacharrya parameter for the non-symmetric SatBP decoder also can be made as small as desired by choosing $K$ large enough. We first consider a fixed computation tree and then average over the tree ensemble. 

We begin with an operational description of symmetrization.
Consider a fixed tree $\sf{T}$ of depth $\ell.$
Let $Y$ denote the vector of received values associated to the variable nodes
under the all-zero codeword assumption.
In addition, for each variable node we assume an independent random variable uniformly
distributed on $[0,1].$  We denote the vector of these variables by $Z.$
Now, the node operations correspond to BP except that outgoing messages from 
the variable nodes are magnitude saturated at $K.$
Independent random variables are used for the flipping operation at each node, where
the flipping probabilty is determined by density evolution.
If the outgoing message has magnitude $K$ then its sign is flipped if $Z_v < \lambda_v$
where $\lambda_v$ is the appropriate flipping probability.

The distribution of the outgoing message $z$ is $S^{(\ell)}_\sym{K}(\Ldens{c}, \Delta_0)).$
Let us consider the conditional distribution $p(z\mid Y,Z).$
We obtain $S^{(\ell)}_\sym{K}(\Ldens{c}, \Delta_0))$ by averaging over $Y$, $Z$ and the code ensemble.
Let $A_K$ denote the event that no flipping occurs.
$
p(z\mid Y ) = p(z\mid Y,A_K) p(A_K) + p(z\mid Y,\bar{A}_K)(1-p(A_K)).
$
Averaging over $Y$ and re-writing we obtain
$$
p(z\mid A_K) = \big(p(z ) -  p(z\mid \bar{A}_K)(1-p(A_K))\big)/p(A_K).
$$
Now $p(z\mid A_K)$ is the distribution of the non-symmetric SatBP decoder.
Intuitively one expects $p(z\mid \bar{A}_K)$ to be inferior (higher probability of error,
larger Battacharyya parameter) to $p(z\mid {A}_K),$ but this appears difficult to prove.

Let us compute the probability of $A_K$. Let the received LLR magnitude of a variable node $v$ be $z \geq K$. 
The probability with which we flip the bit is such that the final error probability is equal to $\frac{e^{-K}}{1+e^{-K}}$. For received LLR magnitude of $z$, the probability that it is received correctly is $\frac1{1+e^{-z}}$. As a consequence we get,
$
\frac{e^{-K}}{1+e^{-K}} = \lambda_v \frac{1}{1 + e^{-z}} + (1 - \lambda_v) \frac{e^{-z}}{1 + e^{-z}},
$
where $\lambda_v$ is the flipping probability of variable node $v$. Solving we get $\lambda_v = \frac{e^{-K}}{1 + e^{-K}} \frac{1 - e^{-z + K}}{1 - e^{-z}} \leq \frac{e^{-K}}{1 + e^{-K}}.$ Thus the probability that a variable node, with a received LLR magnitude greater than $K$, is not flipped is at least equal to $\frac1{1 + e^{-K}} \geq 1 - e^{-K}$. Hence, we get $p(A_K) \ge (1-e^{-K})^{|V(\sf{T})|} \ge  1-e^{-K}{|V(\sf{T})|}$ where ${|V(\sf{T})|}$ is the number of variable nodes in the tree.
From the above analysis we have the following lemma.
\begin{lemma}[SatBP Decoder versus Symmetrized SatBP]\label{lem:clippeddecodervssymclipping}
For any $\epsilon > 0$ and $\ell \in \naturals$, there exists a $K$ large enough
such that
$
\batta(S^{(\ell)}_K(\Ldens{c}, \Delta_0)) \leq \frac1{1-\epsilon}\batta(S^{(\ell)}_\sym{K}(\Ldens{c}, \Delta_0)).
$
\end{lemma}

\section{Stability Analysis}
An important part of the asymptotic analysis of LDPC codes involves the analysis of the convergence of DE to a zero error state.  In this section we analyze the stability of the SatBP.
We begin with some necessary conditions.

For stability of the zero error condition there must exist a positive invariant set
of zero error distributions, i.e,  a subset $\mathcal{S}$ of distributions
so that $\perr{(\Ldens{s})} =0$ for all $\Ldens{s} \in \mathcal{S}$ and
$S_{K}(\Ldens{c}, \Ldens{s}) \in \mathcal{S}.$
Existence of $\mathcal{S}$ follows easily from the compactness of the 
space of densities and continuity of DE.

\begin{lemma}\label{lem:necessarysupport}
Assume the channel $\Ldens{c}$ has support at $-L,$ $L>0.$
In an irregular ensemble with minimum variable degree $d_l$
the support of all densities in $\mathcal{S}$ must lie in 
$[L/(d_l-2),\infty).$
\end{lemma}
\begin{IEEEproof}
Assume  $\Ldens{a}^{(0)}\in \mathcal{S}$ has support at $z<L/(d_l-2).$
Then $\Ldens{b}^{(0)}$ has support on $[-\infty,z].$ 
Hence $\Ldens{a}^{(1)}$ has support on $[-\infty,z_1]$ 
where $z_1 = (d_l-1) z - L< L - 2\delta$
where $\delta = L-(d_l-2) z.$
By induction we have $\Ldens{a}^{(k)}$ has support on $[-\infty,L-2^k\delta]$ 
and for $k$ large enough the probability of error is positive.
\end{IEEEproof}

\subsection{Failure of Stability with Degree Two}

From Lemma \ref{lem:necessarysupport} we have immediately
\begin{lemma}
In an irregular ensemble with $\lambda_2>0$ no invariant $\mathcal{S}$
exists for any value of $K<\infty$ unless the channel is the BEC.
\end{lemma}
\begin{IEEEproof}
If $d_l = 2$ and the channel is not the BEC and hence has support on $(-\infty,0),$
then Lemma \ref{lem:necessarysupport}
shows that there can be no positive invariant zero-error set of distributions with support
on $[-K,K].$
\end{IEEEproof}

In the case of the BEC it can be seen that saturated DE matches unsaturated DE except that the mass
at $+\infty$ in unsaturated DE is not placed at $+K.$  Hence, stability is unaffected by saturation.

If the channel has infinite support, then there is no possibility of stability under saturation.
The condition on the  finite channel support is given later in the section on stability with degree $\geq 3$.

\subsection{Near Stability}\label{sec:nearstab}
The stability analysis of standard irregular ensembles under BP decoding
rests on on the results 
$\batta(\Ldens{c}\vconv \lambda{(\Ldens{a})})  = \batta{(\Ldens{c})} \lambda{(\batta{(\Ldens{a})})}$ and
$\batta{(\rho{(\Ldens{a})})} \le 1 - \rho(1-\batta{(\Ldens{a})})\label{eqn:DEcnodebatta}$.
The first equation continues to hold without symmetry of 
$\Ldens{a}$ or $\Ldens{c}.$
The check node inequality, however, does not necessarily hold without symmetry.
We have, however, a substitute with a slight variation on a result from \cite{BRU07}:
for any L-densities $\Ldens{a}$ and $\Ldens{b}$ we have
$
\batta{(\Ldens{a} \cconv \Ldens{b})} \le
\batta{(\Ldens{a})} + \batta{( \Ldens{b})}\,.
$
This result holds for a wide range of check node update operations including
BP and the min-sum decoder.  To incorporate saturation into the analysis based on
the Bhattacharyya parameter we have the inequality
$
\batta{(\SatLdens{a}_K)}\le \batta{(\Ldens{a})} + e^{-K/2}.
$
Note that because of saturation, the minimum value of the Battacharyya parameter is equal to $e^{-K/2}$.

{\em Minimum variable node degree equal to 2:}
Let $\Ldens{a}^{(k)}$ be any $L$-density which need not be symmetric.
Consider an irregular ensemble and assume $\rho'(1) \batta{(\Ldens{a}^{(k)}})<1.$
Using the notation $\bar{\lambda}_2 = 1 - \lambda_2$ we upper bound $
\batta(\Ldens{a}^{(k\!+\!1)})$ for $\Ldens{a}^{(k\!+\!1)} = S_K(\Ldens{c}, \Ldens{a}^{(k)})$ by
\begin{align*}
\lambda_2  \batta(\Ldens{c}) \rho'(1) \batta{(\Ldens{a}^{(k)})} 
\!+\!\bar{\lambda}_2 \batta(\Ldens{c}) (\rho'(1) \batta{(\Ldens{a}^{(k)}}))^2
\!+\! e^{\frac{-K}2}
\end{align*}
Let $\xi>0$ satisfy 
\[
\eta := \lambda_2  \batta(\Ldens{c}) \rho'(1) + \bar{\lambda}_2\batta(\Ldens{c}) (\rho'(1))^{{2}} \xi < 1
\]
and assume $\xi \ge \batta{(\Ldens{a}^{(k)}})$ and $K$ is large enough so that $\xi \ge \frac{2-\eta}{1-\eta} e^{-K/2}.$
Then there exists $N$ so that for $n\ge N$ we have
$
\batta{(\Ldens{a}^{(n)})}
\le
\frac{1}{1-\eta} e^{-K/2}.
$

{\em Minimum variable node degree equal to 3:}
Let us now assume that the minimum variable node degree is 3. Following the previous notation, we upper bound $\batta{(\Ldens{a}^{(k+1)}})$ by
\begin{align*}
\lambda_3  \batta(\Ldens{c}) (\rho'(1) \batta{(\Ldens{a}^{(k)})})^2 + \bar{\lambda}_3\batta(\Ldens{c}) (\rho'(1) \batta{(\Ldens{a}^{(k)}}))^3
+ e^{\frac{-K}2}
\end{align*}
Let $\xi$ be the positive solution to
$
 \lambda_3  \batta(\Ldens{c}) \rho'(1)^2 \xi
+\bar{\lambda}_3\batta(\Ldens{c}) \rho'(1)^3 \xi^2
= 1/2\,.
$
Assume $K$ large enough so that $2e^{-K/2} < \xi.$ 
A little algebra shows that if $\batta{(\Ldens{a}^{(0)})} \le  \xi,$ then there exists $N$ so that for all $n\ge N$  we have
\begin{align}
\batta{(\Ldens{a}^{(n)})} & \le 2e^{-K/2}, \quad \quad \batta{(\Ldens{b}^{(n)})}  \le  2\rho'(1) e^{-{K/2}}\label{vnodecheckallbnd}
\end{align}
where the second inequality follows from the first and the additive bound on Bhattacharrya at check node.

This analysis can not show convergence to zero error although it can be used
to show convergence to relatively small error rate.  This is true even in the presence
of degree two variable nodes, where zero error stability is not possible.
For degree three and higher stability can be shown, but a refined analysis is needed.

\subsection{Stability with Minimum Variable Degree Equal to Three}
In this section we consider irregular ensembles where the minimum variable node degree
is at least three.  We generalize the standard stability analysis by separating out the
saturated probability mass and tracking it through the variable node and check node updates.
For simplicity we shall restrict to right regular ensembles.
If the check node degree is $d_r$ then $K'$ is the magnitude of an outgoing message
all of whose incoming messages have magnitude $K.$  Although we focus on BP-like decoding
our analysis applies to other algorithms such as min-sum, in which case we have $K'=K.$
In general, if $K_1,...,K_d$ are the magnitudes of incoming messages then we assume that
the outgoing magnitude satisfies
$
K' \le \min_i \{ K_i \}
$
and
$
e^{-K'} \le \sum_i e^{-K_i}.
$
Both conditions are satisfied by BP and min-sum.
Messages entering a check node update $\Ldens{a}$ have the form 
$
\Ldens{a} = \gamma D(p,K) +\bar{\gamma}\Ldens{m},
$
where $\Ldens{m}$ is supported on $(-K,K)$ and has total mass $1.$
Messages entering a variable node update $\Ldens{b}$ have the form 
$
\Ldens{b} = \gamma D(p,K') +\bar{\gamma}\Ldens{m},
$
where $K'<K$ is the outgoing magnitude at a check when all incoming magnitudes equal $K$
and $\Ldens{m}$ is supported on $(-K',K').$ 
From above, we have $e^{-K'} \le (d_r-1) e^{-K}$ so that
$K' \ge K -\ln(d_r-1).$ Furthermore, we choose $K$ large enough so that $2K' > K$.

In the sequel we assume that the support of the channel $\Ldens{c}$ is $(-K'', K'')$. The condition on $K''$ is that $K'' \leq 2K' \!-\! K $. 
The proof of the two statements below can be found in the upcoming paper \cite{KRI14}.
\subsubsection{Variable Node Analysis}
Let us assume a variable node of degree $d+1$ and incoming density
$
\Ldens{b} = \gamma D(p,K') +\bar{\gamma}\Ldens{m}.
$
The outgoing density from the variable node has the form 
$
\Ldens{a} =\gamma' D(p',K) +\bar{\gamma}'\Ldens{m'}.
$
Then we have
\begin{align*}
\bar{\gamma}' \batta(\Ldens{m'})
 \le&
 \batta{(\Ldens{c})} d e^{-(d-1)(K/2-\ln c)} (\bar{\gamma}\batta(\Ldens{m})) \,
\\& + (\gamma p)   e^{-(d/2-1)(K-\ln c)} e^{d/2\ln (3e)}(1 + 2\dr), \\
{\gamma}' p'
\le &
\, e^{-K/2} \batta{(\Ldens{c})}  d e^{-(d-1)(K/2-\ln c)} (\bar{\gamma}\batta(\Ldens{m}))
\\ & 
 +2^{d} (\gamma p)^{\lfloor (d+2)/2 \rfloor}\,.
\end{align*}
\subsubsection{Check Node Analysis}
Let us assume a right regular ensemble with check degree $d+1.$
Let us represent the density entering the check node as
$\gamma D(p,K) + \bar{\gamma}\Ldens{m}$
where $\Ldens{m}$ is a density supported on $(-K,K).$
Then the density emerging out of the check node is given by 
$ \gamma' D(p',K_d) + \bar{\gamma'}\Ldens{m}' \triangleq (\gamma D(p,K) + \bar{\gamma}\Ldens{m})^{\cconv d}
$, where $K_d$ is the magnitude,
which satisfies $K-\ln d \le K_d \le K$ and support of $\Ldens{m}$ is also $(-K_d, K_d)$. Then we have,
$\bar{\gamma'}\batta(\Ldens{m}')  \leq \bar{\gamma}\batta{(\Ldens{m})}( \gamma C d^3 + d)$. Also, we note that $p' = p_d \le d p.$ 
\subsubsection{Proof of Stability for Degree Three+}\label{sec:stabdegthree+}
Let us assume that the minimum variable node degree is at least three
and a right regular degree $d_r.$ 
In view of \eqref{vnodecheckallbnd} we may assume
$\batta{(\Ldens{a}^{(n)})}  \le 2e^{-K/2}$  which implies
$\batta{(\Ldens{b}^{(n)})}  \le  2(d_r-1) e^{-{K/2}}$
for all $n \ge 0.$
We write $\Ldens{a}^{(n)} = \gamma D(p,K) + \bar{\gamma}\Ldens{m}$
and now have the bound (for all $n$)  $\gamma pe^K \le 2.$
Set  $c=2(d_r-1)$ and from the check node analysis above and the bound  $\gamma pe^K \le 2$ we get $\gamma C \le 2 (d_r-1) + 1 + \sqrt{d_r-1}.$ Abusing the notation, set $C = 2 (d_r-1) + 1 + \sqrt{d_r-1}.$ 
Then $\Ldens{b}^{(n)} = \gamma' D(p',K_{d_r-1}) +  \bar{\gamma}'\Ldens{m}'$ denote the density emerging from the check node update where the support of $\Ldens{m}'$ is $(-K_{\dr-1}, K_{\dr-1})$.  Then from the check node analysis above we get
\begin{align*}
\gamma' p' &\le  (d_r-1) \gamma p \\
 \bar{\gamma} \batta{(\Ldens{m}')} & \le  \bar{\gamma}\batta{(\Ldens{m})}( \gamma C (\dr-1)^3 + \dr-1). 
\end{align*}
We choose $a,b,c,d>0$ arbitratily small.  Then for all $K$ large enough and for all $d \ge 2$ 
and $d$ less then the maximum variable node degree, we have
\begin{align*}
(d_r-1)  2^{d} (2e^{-K/2})^{\lfloor d/2 \rfloor} & \le a \\
\batta{(\Ldens{c})}  d e^{-(d-1)(K/2-\ln c)} (d_r-1)C & \le b \\
 e^{-(\frac{d}2-1)(K-\ln c)} e^{\frac{d}2\ln (3e)}(1 + 2\dr)  2(d_r-1)e^{-K/2}& \le c \\
  \batta{(\Ldens{c})} d e^{-(d-1)(K/2-\ln c)} (d_r-1) C  & \le d 
\end{align*}
We now obtain for $n \ge 0,$
\begin{align*}
\begin{bmatrix} e^{K/2}  \gamma    p  \\  \bar{\gamma} \batta{(\Ldens{m})} \end{bmatrix}^{(n+1)}
\le
\begin{bmatrix}
 a        & b \\
 c & d
\end{bmatrix}
\begin{bmatrix} e^{K/2}  \gamma    p  \\  \bar{\gamma} \batta{(\Ldens{m})} \end{bmatrix}^{(n)}
\end{align*}
which  is a stable system for suitable $a,b,c,d.$ It is now easy to see that the bit error rate goes to zero since
the probability mass on $[-K,K)$ goes to zero as $n$ becomes large.

\section{Discussion}
Consider a channel family, BMS($\ent$), ordered by entropy and let $\ent^{\BPsmall}$ denote the BP threshold when transmitting over this channel family using a $(\lambda, \rho)$ LDPC ensemble such that the variable degree is at least three. 
 Note that when we consider a channel from BMS($\ent$), we first (symmetric) saturate the channel to an appropriate finite support. One can show that as the support goes to $(-\infty, \infty)$, the saturated channel BP threshold goes to $\ent^{\BPsmall}$ (cf. \cite{KRI14}).  The above analysis shows that, for $K$ large enough, the SatBP is successful for any channel with entropy arbitrarily close to channel corresponding to $\ent^{\BPsmall}$.

An interesting open question is to understand the scaling of $K$ with the gap to capacity $\delta$. From Lemma~\ref{lem:distsymclippedandBP} we see that one would require the scaling of number of iterations wrt $\delta$ for unsaturated BP which is conjectured to scale like $O(1/\delta)$.

\section{Acknowledgements}
SK would like to thank R. Urbanke for useful discussions.
\bibliographystyle{IEEEtran}
\bibliography{lth,lthpub,extras}

\end{document}